%% file: asrHO_arxiv.tex
\documentclass{llncs}

\usepackage{verbatim}
\usepackage{amssymb}
\usepackage{textcomp}
\usepackage{html}
\usepackage{amsmath}

\input{macros}

\title{An Approach to Higher-Order Assertion-based \\ Debugging
       of Higher-Order (C)LP Programs
  \thanks{Research supported in part by projects EU
    FP7 318337 \emph{ENTRA}, Spanish MINECO TIN2012-39391
    \emph{StrongSoft} and TIN2008-05624 \emph{DOVES}, and Comunidad de
    Madrid TIC/1465 \emph{PROMETIDOS-CM}.}
}

\author{Nataliia Stulova,\inst{1} Jos\'{e} F. Morales,\inst{1}
 and Manuel V. Hermenegildo\inst{1,2}}

\institute{IMDEA Software Institute, Madrid, Spain
\and School of Computer Science, T. U. Madrid (UPM), Spain}

\begin{document}
\maketitle


\begin{abstract}
  Higher-order constructs extend the expressiveness of first-order
  (Constraint) Logic Programming ((C)LP) both syntactically and
  semantically. At the same time assertions have been in use for some
  time in (C)LP systems helping programmers detect errors and
  validate programs. However, these assertion-based extensions to
  (C)LP have not been integrated well with higher-order to date. This
  paper contributes to filling this gap by extending the
  assertion-based approach to error detection and program validation
  to the higher-order context within (C)LP. We propose an extension
  of properties and assertions as used in (C)LP in order to be able
  to fully describe arguments that are predicates. The extension
  makes the full power of the assertion language available when
  describing higher-order arguments. We provide syntax and semantics
  for (higher-order) properties and assertions, as well as for
  programs which contain such assertions, including the notions of
  error and partial correctness. We also discuss several alternatives
  for performing run-time checking of such programs.
\end{abstract}


\section{Introduction}

Higher-order programming adds flexibility to the software development
process. Within the (Constraint) Logic Programming ((C)LP) paradigm,
Prolog has included higher-order constructs since the early days, and
there have many other proposals for combining the first-order kernel
of (C)LP with different higher-order constructs (see,
e.g.,~\cite{warren-hiord,hiord-naish,ChenKiferWarren93,%
NadathurMiller98,ciao-hiord,daniel-phd}). Many of these proposals
are currently in use in different (C)LP systems and have been found
very useful in programming practice, inheriting the well-known
benefits of code reuse (templates), elegance, clarity, and
modularization.

A number of extensions have also been proposed for (C)LP in order
to enhance the process of error detection and program validation. In
addition to the use of classical strong
typing~\cite{goedel,mercury-jlp-short}, a number of other approaches
have been proposed which are based on the dynamic and/or static
checking of user-provided, optional \emph{assertions}~\cite{DNM88,%
assert-lang-ws,aadebug97-informal-short,%
BDM97,assrt-theoret-framework-lopstr99,%
DBLP:conf/discipl/Lai00,ciaopp-sas03-journal-scp-short,%
testchecks-iclp09}.
In practice, different aspects of the model
of~\cite{assrt-theoret-framework-lopstr99,%
ciaopp-sas03-journal-scp-short}
have been incorporated in a number of widely-used (C)LP systems, such
as Ciao, SWI, and XSB~\cite{hermenegildo11:ciao-design-tplp,%
xsb-journal-2012,DBLP:journals/tplp/MeraW13}.
A similar evolution is represented by the soft/gradual typing-based
approaches in functional programming and the contracts-based
extensions in object-oriented
programming~\cite{cartwright91:soft_typing-short,TypedSchemeF08-short,%
clousot,lamport99:types_spec_lang,DBLP:journals/fac/LeavensLM07}.

These two aspects, assertions and higher-order, are not independent.
When higher-order constructs are introduced in the language it becomes
necessary to describe properties of arguments of predicates that are
themselves also predicates.  While the combination of contracts and
higher-order has received some attention in functional
programming~\cite{DBLP:conf/icfp/FindlerF02,%
DBLP:journals/toplas/DimoulasF11},
within (C)LP the combination of higher-order with the previously
mentioned assertion-based approaches has received comparatively little
attention to date. Current Prolog systems simply use basic atomic
types (i.e., stating simply that the argument is a \texttt{pred},
\texttt{callable}, etc.) to describe predicate-bearing variables.
The approach of~\cite{BeierleKloosMeyer99} is oriented to meta
programming. It allows describing meta-types but there is no notion of
directionality (modes), and only a single pattern is allowed per
predicate.

This paper contributes to filling the existing gap between
higher-order and assertions in (C)LP.
Our starting point is the Ciao assertion
model~\cite{assrt-theoret-framework-lopstr99,%
  ciaopp-sas03-journal-scp-short}, since, as mentioned before, it has
been adopted at least in part in a number of the most popular (C)LP
systems.
After some preliminaries and notation (Section~\ref{sec:Prel-Notat})
we start by extending the traditional notion of programs and
derivations in order to deal with higher-order calls and recall and
adapt the notions of first-order conditional literals, assertions,
program correctness, and run-time checking to this type of
derivations (Section~\ref{sec:fo-ho}). This part allows us to revisit
the traditional model in this new, higher-order context, while
introducing a different formalization than the original one
of~\cite{assrt-theoret-framework-lopstr99}. This formalization, which
will be used throughout the paper, is more compact and gathers all
assertion violations as opposed to just the first one, among other
differences.
We then define an extension of the properties used in assertions and
of the assertions themselves to higher-order, and provide
corresponding semantics and results (Section~\ref{sec:ho-ho}).


\section{Preliminaries and Notation}
\label{sec:Prel-Notat}

We recall some concepts and notation from standard (C)LP theory.
We denote by $\VS$, $\FS$, and $\PS$ the set of variable, function,
and predicate symbols, respectively. Variables start with a capital
letter. Each $p \in \PS$ and $f \in \FS$ is associated to a natural
number called its \emph{arity}, written $\ar{p}$ or $\ar{f}$.
The set of terms $\TS$ is inductivelly defined as follows: 
$\VS \subset \TS$, if $f \in \FS$ and $t_1, \ldots, t_n \in \TS$
then $f(t_1,\ldots,t_n) \in \TS$ where $\ar{f}=n$.
An \emph{atom} has the form $p(t_1,...,t_n)$ where 
$p \in PS$, $\ar{p} = n$, and $t_1,...,t_n \in TS$. 
A \emph{constraint} is essentially a conjunction of expressions built
from predefined predicates (such as term equations or inequalities
over the reals) whose arguments are constructed using predefined
functions (such as real addition).
A \emph{literal} is either an atom or a constraint. 
A \emph{goal} is a finite sequence of literals.
A \emph{rule} is of the form $H \mbox{\tt :-} B$ where $H$, the
\emph{head}, is an atom and $B$, the \emph{body}, is a possibly empty
finite sequence of literals.
A \emph{constraint logic program}, or \emph{program}, is a finite set
of rules.

We use $\sigma$ to represent a variable renaming and $\sigma(X)$ to
represent the result of applying the renaming $\sigma$ to some
syntactic object $X$ (a term, atom, literal, goal, etc.).
The \emph{definition} of an atom $A$ in a program, $\defn{A}$, is the
set of variable renamings of the program rules such that each renaming
has $A$ as a head and has distinct new local variables.
We assume that all rule heads are normalized, i.e., $H$ is of the form  
$p(X_1,...,X_n)$ where the $X_1,...,X_n$ are distinct free
variables. This is not restrictive since programs can always be
normalized, and it facilitates the presentation. However, for
conciseness in the examples
we sometimes use non-normalized programs.
Let $\overline{\exists}_L \theta$ be the constraint $\theta$
restricted to the variables of the syntactic object $L$. We denote
\emph{constraint entailment} by $\models$, so that $\theta_1\models
\theta_2$ denotes that $\theta_1$ entails $\theta_2$. In such case we
say that $\theta_2$ is \emph{weaker} than $\theta_1$.

For brevity, we will assume in the rest of the paper that we
are dealing with a single program, so that all sets of rules, etc.\
refer to that implicit program and it is not necessary to refer to it
explicitly in the notation.


\subsection{Operational Semantics}

The operational semantics of a program is given in terms of its
``derivations,'' which are sequences of reductions between ``states.'' 
A \emph{state} $\state{G}{\theta}$ consists of a goal $G$ and a
constraint store (or \emph{store} for short) $\theta$.
We use :: to denote concatenation of sequences and we assume for
simplicity that the underlying constraint solver is complete.
We use $\reduction{S}{S'}$ to indicate that a reduction can be applied
to state $S$ to obtain state $S'$. Also, $\reductionStar{S}{S'}$
indicates that there is a sequence of reduction steps from state $S$
to state $S'$.
We denote by $\nthstate{D}{i}$ the $i$-th state of the derivation. As
a shorthand, given a non-empty derivation $D$, $\currstate{D}$ denotes 
the last state.
A \emph{query} is a pair $(L,\theta)$, where $L$ is a literal and
$\theta$ a store, for which the (C)LP system starts a computation from
state $\state{L}{\theta}$.
The set of all derivations from the query $Q$ is denoted
$\derivations{Q}$.
The observational behavior of a program is given by its ``answers'' to
queries. A finite derivation from a query $(L,\theta)$ is
\emph{finished} if the last state in the derivation cannot be
reduced. 
Note that $\derivations{Q}$ contains not only finished derivations but
also all intermediate derivations from a query.
A finished derivation from a query $(L,\theta)$ is
\emph{successful} if the last state is of the form
$\state{\emptyGoal}{\theta'}$, where $\emptyGoal$ denotes the empty
goal sequence. In that case, the constraint $\bar{\exists}_{L}\theta'$
is an \emph{answer} to $S$.
We denote by $\answers{Q}$ the set of answers to a query $Q$.
A finished derivation is \emph{failed} if the last state is not of the
form $\state{\emptyGoal}{\theta}$.
A query $Q$ \emph{finitely fails} if $\derivations{Q}$ is
finite and contains no successful derivation.


\section{First-order Assertions on Higher-order Derivations}
\label{sec:fo-ho}


\subsection{Higher-order Programs and Derivations}

We start by extending the definition of program, state reduction,
and derivations in order to deal with the syntax and semantics of
higher-order calls.
 
\begin{definition}[Higher-order Programs]
  \emph{Higher-order programs} are a generalization of
  \emph{constraint logic programs} where:
  \begin{itemize}
  \item The set of literals $\LS$ is extended to include
    \emph{higher-order literals} $X(t_1,\ldots,t_n)$, where $X \in
    \VS$ and the $t_i \in \TS$.
  \item The set of terms $\TS$ is extended so that $\PS \subset \TS$
    (i.e., predicate symbols $p$ can be used as constants).
  \end{itemize}
\end{definition}
In the following we assume a simple semantics where when a call to a
higher-order literal $X(t_1,\ldots,t_n)$ occurs, $X$ has to be
constrained to a predicate symbol in the store:\footnote{This is also
  the most frequent semantics in current systems. Other alternatives,
  such as residuation~\cite{HASSAN93} (delays), predicate enumeration,
  etc.\ can also be used, requiring relatively straightforward
  adaptations of the model proposed.}

\begin{definition}[Reductions in Higher-order Programs]
  \label{def:ho-reductions}
  A state $S = \state{L::G}{\theta}$ where $L$ is a literal can be
  \emph{reduced} to a state $S'$, denoted $\reduction{S}{S'}$, as
  follows:
  \begin{enumerate}
  \item 
    If $L$ is a constraint and $\theta \land L$ is satisfiable,
    then $S' = \state{G}{\theta \land L}$.
  \item 
    If $L$ is an atom of the form $p(t_1,\ldots,t_n)$, for some
    rule $(L \mbox{\tt :-} B)$ $\in \defn{L}$, then $S' =
    \state{B::G}{\theta}$.
  \item If $L$ is of the form $X(t_1,\ldots,t_n)$, then 
    $S' = \state{G'}{\theta}$ where:
    \[
    G' = \left\{
      \begin{array}{lr}
        p(t_1,\ldots,t_n) :: G  
      & ~~~~\textrm{if}~ \exists p \in \PS \wedge \theta \models 
        (X = p) \wedge \ar{p} = n 
      \\
        \epsilon_{uninst\_call} 
      & \textrm{otherwise}
      \end{array}
      \right.
    \]
  \end{enumerate}
\end{definition}

The concepts of answers and of finished and successful derivations
carry over without change to this notion of higher-order derivations.
The notion of (finitely) failed derivation is extended as follows:
\begin{definition}[(Finitely) Failed Derivation]
  A finished derivation from a query $(L,\theta)$ is \emph{failed} iff
  its last state is not of the form $\state{\emptyGoal}{\theta'}$ or
  $\state{\epsilon_{uninst\_call}}{\theta}$.
\end{definition}
Finally, we introduce the concept of \emph{floundered} derivations:
\begin{definition}[Floundered Derivation]
  A finished derivation from a query $(L,\theta)$ is \emph{floundered}
  iff its last state is of the form 
  $\state{\epsilon_{uninst\_call}}{\theta}$.
\end{definition}
%


\subsection{First-order \emph{Pred} Assertions}
\label{sec:assertions}

Assertions are linguistic constructions for expressing properties of
programs. They are used for detecting deviations of the program
behavior (symptoms) with respect to such assertions, or to ensure
that no such deviations exist (correctness).
Herein,
we will use the \emph{pred} assertions
of~\cite{assert-lang-disciplbook-short},
given that they are the most frequently used assertions in practice,
and they subsume the other assertion schemas in that language.
Thus, in the following we will
use simply the term assertion to refer to a pred assertion.
Assertions allow specifying certain conditions on the
constraint store that must hold at certain points of program
derivations. In particular, they allow stating sets of
\emph{preconditions} and \emph{conditional postconditions} for a
given predicate.
A set of assertions for a predicate is of the form:
\begin{small}
  \[
  \begin{array}{l}
    \kbd{:- pred } Head \kbd{ : } Pre_1 \kbd{ => } Post_1 \kbd{.}
  \\
    \ldots
  \\
    \kbd{:- pred } Head \kbd{ : } Pre_n \kbd{ => } Post_n \kbd{.}
  \end{array}
  \]
\end{small}
where $Head$ is a normalized atom that denotes the predicate that the
assertions apply to, and the $Pre_i$ and $Post_i$ refer to the
variables of $Head$. We assume that variables in assertions are renamed
such that the $Head$ atom is identical for all assertions for a given
predicate.  A set of assertions as above states that in any execution
state $\state{Head :: G}{\theta}$
at least one of the $Pre_i$ conditions should hold, and that, given
the $(Pre_i,Post_i)$ pair(s) where $Pre_i$ holds, then, if $Head$
succeeds, the corresponding $Post_i$ should hold upon success.
The following example illustrates the basic concepts involved:

\begin{example}
  \label{ex:qsortasrt}
  The procedure \skbd{qsort(A,B)} is the usual one 
  that relates lists \skbd{A} and their sorted versions \skbd{B}. 
  The following assertions:
  \begin{small}
    \[
    \begin{array}{l}
      \kbd{:- pred qsort(A,B) : list(A) => (sorted(B), list(B)).}
    \\
      \kbd{:- pred qsort(A,B) : list(B) => (permutation(B,A), 
      list(A)).}
    \end{array}
    \]
  \end{small}
state that (restrict the meaning of \kbd{qsort} to):
  \vspace{-0.8em}
  \begin{itemize}
    \item \skbd{qsort(A,B)} should be called either with \skbd{A}
      constrained to a list or with \kbd{B} constrained to a list;
    \item if \skbd{qsort(A,B)} succeeds when called with \skbd{A}
      constrained to a list then on success \skbd{B} should be a
      sorted list; 
    \item if \skbd{qsort(A,B)} succeeds when called with \skbd{B}
      constrained to a list then on success \skbd{A} should be a list
      which is a permutation of \skbd{B}.
  \end{itemize}
\end{example}


\subsection{Conditions on the Constraint Store}

The conditions on the constraint store used in assertions are
specified by means of special literals (e.g., \skbd{list(A)},
\skbd{sorted(B)}, \skbd{list(B)}, and \skbd{permutation(B,A)}
in the previous example) that we will herein call \emph{prop}
literals. More concretely, we assume the $Pre_i$ and $Post_i$ to be
DNF formulas of such literals.

We also assume that for each prop literal $L_p$ used in some
assertion there exists a corresponding predicate $p$ defining it.
Then, we can define the meaning of prop literals as follows:

\begin{definition}[Meaning of a \emph{Prop} Literal] 
  The meaning of a prop literal $L_p$ defined by predicate $p$,
  denoted $|L_p|$, is the set of constraints given by
  $\answers{(L_p,true)}$.
\end{definition}

Intuitively, the meaning of prop literals is the set of ``weakest''
constraints for which the literal holds:

\begin{example}
  \label{ex:props}
  Prop literals \skbd{list/1} and \skbd{sorted/1} can be defined by:
    \[
    \begin{array}{lll}
      \skbd{list([]).}~~ 
    & ~~\skbd{sorted([]).}~~ 
    & ~~\skbd{sorted([\_]).}
    \\
      \skbd{list([\_|L]) :- list(L).}~~
    & ~~\skbd{sorted([X,Y|L]) :- }
    & \skbd{X =< Y, sorted([Y|L]).}
    \end{array}
    \]
  Then, their meaning is given by:
  \[
  \begin{array}{c}
    |list(A)| = \{A=[], A=[B|C] \wedge list(C)\}, \textrm{ and}
  \\
    |sorted(A)| = \{A=[], A=[B], A=[B,C|D] \wedge B\leq C
    \wedge E = [C|D] \wedge sorted(E)\}.
  \end{array}
  \]
\end{example}

The following definition from~\cite{assrt-theoret-framework-lopstr99}
defines when the condition represented by a prop literal (defined by a
program predicate) holds for a given store:

\begin{definition}[Succeeds Trivially]
  \label{def:succeeds-triv}
  A prop literal $L$ \emph{succeeds trivially} for $\theta$, denoted
  $\theta \Rightarrow_P L$, iff $\exists \theta'\in
  \answers{(L,\theta)} \mbox{ such that } \theta\models\theta'$.
  A DNF formula of prop literals succeeds trivially for $\theta$ if
  all of the prop literals of at least one conjunct of the formula
  succeeds trivially.
\end{definition}

Intuitively, a prop literal $L$ succeeds trivially if $L$ succeeds for
$\theta$ without adding new ``relevant'' constraints to $\theta$:

\begin{example} 
\label{ex:props-more}
  Consider prop literals $list(A)$ and $sorted(B)$ and the predicate
  definitions of Example~\ref{ex:props}:
\vspace{-0.75em}
  \begin{itemize}
  \item
  Assume that $\theta= (A=f)$.
  Since $\forall \theta'\in |list(A)|$ : $\theta \not\models \theta'$,
  as we would expect, $\theta \not\Rightarrow_P list(A)$.
  
  \item
  Assume now that $\theta= (A=[\underline{~~}|Xs])$.  Though $A$ is
  compatible with a list, it is not actually a (nil terminated)
  list. Again in this case $\forall \theta'\in |list(A)|$ : $\theta
  \not\models \theta'$ and thus again $\theta \not\Rightarrow_P
  list(A)$. The intuition behind this is that we cannot guarantee that
  $A$ is actually a list given $\theta$, since a possible instance of
  $A$ in $\theta$ is $A=[\underline{~~}|f]$, which is clearly not a
  list.
  
  \item
  Finally, assume that $\theta=(A=[B]\wedge B=1)$.  In such case
  $\exists \theta'= (A=[B|C]\wedge C=[])$ such that\ $\theta \models
  \theta'$ and $\exists c = (B = 1)$ such that\ $(c\wedge \theta'
  \not\models false)\wedge (\theta'\wedge c\models \theta$).  Thus, in
  this last case $\theta \Rightarrow_P list(A)$.
  \end{itemize}
\end{example}

This means that we are considering prop literals as
\emph{instantiation}
checks~\cite{prog-glob-an,assert-lang-disciplbook-short}: they are
true iff the variables they check for are at least as constrained as
their predicate definition requires.

\begin{definition}[Test Literal]
\label{def:test-literal}
  A prop literal $L$ is a \emph{test} iff $\forall \theta$ either
  $\theta \Rightarrow_P L$ or $(L,\theta)$ finitely fails.
\end{definition}


\subsection{First-order Assertion Conditions and their Semantics}

We represent the different checks on the constraint store imposed by a
set of assertions as a set of \emph{assertion conditions} as follows.

\begin{definition}[Assertion Conditions for a Predicate]
  \label{def:assrt-cond}
  Given a predicate represented by a normalized atom $Head$, if the
  corresponding set of assertions is $\A = \{A_1 \ldots A_n\}$, with
  $A_i = ``\texttt{:- pred } Head \texttt{ : } Pre_i \texttt{ => }
  Post_i \texttt{.}$'' the set of \emph{assertion conditions} for
  $Head$ is $\{ C_0, C_1, \ldots , C_n\}$, 
  with:
  \[
    C_i = \left\{
    \begin{array}{ll}
      \callsAsr{Head}{\bigvee _{j = 1}^{n} Pre_j}
    & ~~~~i = 0 
    \\
      \successAsr{Head}{Pre_i}{Post_i}
    & ~~~~i = 1..n
    \end{array}
    \right.
  \]
\end{definition}

If there are no assertions associated with $Head$ then the
corresponding set of conditions is empty.
The set of assertion conditions for a program is the union of the
assertion conditions for each of the predicates in the program.
Also, given a single assertion $A_i$ we define its corresponding set
of assertion conditions as $\{ C_0, C_i \}$ (this will be useful in
defining the status of an assertion).

The $\callsAsr{Head}{\ldots}$ conditions encode the checks that the
calls to the predicate represented by $Head$ are within those
admissible by the set of assertions, and we thus call them the
\emph{calls assertion conditions}.  The
$\successAsr{Head_i}{Pre_i}{Post_i}$ conditions encode the checks for
compliance of the successes for particular sets of calls, and we thus
call them the \emph{success assertion conditions}.

\begin{example}
  The assertion conditions corresponding to the predicate assertions
  for \texttt{qsort} in Example~\ref{ex:qsortasrt} are as follows:
  \begin{small}
    \[
    \begin{array}{l}
      \callsAsr{qsort(A,B)}{(list(A),list(B))}
    \\
      \successAsr{qsort(A,B)}{list(A)}{(sorted(B),list(B))}
    \\
      \successAsr{qsort(A,B)}{list(B)}{(permutation(B,A), list(A))}
    \end{array}
    \]
  \end{small}
\end{example}

In order to define the semantics of assertion conditions, we introduce 
the auxiliary partial functions $\textsf{prestep}$ and $\textsf{step}$
as follows:
\begin{align*}
  \textsf{prestep}(L_a,D)=(\theta,\sigma) \equiv &~
  \currstate{D}=\state{L::G}{\theta} \wedge
  \exists \sigma\;L=\sigma(L_a) \\
  \textsf{step}(L_a,D)=(\theta,\sigma,\theta') \equiv &~
    \currstate{D}=\state{G}{\theta'}\;
    \wedge
    \exists i\; \nthstate{D}{i} = \state{L::G}{\theta}
    \wedge \exists \sigma\;L=\sigma(L_a)
\end{align*}

Given a derivation whose current state is a call to $L_a$ (normalized
atom), the $\textsf{prestep}$ function returns the substitution
$\sigma$ for $L_a$, and the constraint store $\theta$ at the
predicate \emph{call} (i.e., just before the literal is reduced).
Given a derivation whose current state corresponds exactly to the
return from a call to $L_a$, the $\textsf{step}$ function returns
the substitution $\sigma$ for $L_a$, the constraint store $\theta$
at the call to $L_a$, and the constraint store $\theta'$ at $L_a$'s
\emph{success} (i.e., just after all literals introduced from the
body of $L_a$ have been fully reduced).
Using these functions, the semantics of our \emph{calls} and
\emph{success} assertion conditions are given by the following
definition:

\begin{definition}[Valuation of an Assertion Condition on a Derivation]
  \label{def:assrt-valuation}
  Given a \emph{calls} or \emph{success} assertion condition $C$, the
  \emph{valuation of $C$ on a derivation $D$}, denoted 
  $\textsf{solve}(C,D)$ is defined as follows: 
  \[
  \begin{array}{rl}
    \textsf{solve}(\callsAsr{L_a}{Pre},D)
    \equiv
    &
    \left( \textsf{prestep}(L_a,D)=(\theta,\sigma) \right) 
    \Rightarrow
    \left(\theta\Rightarrow_P\sigma(Pre)\right)
  \\
    \textsf{solve}(\successAsr{L_a}{Pre}{Post},D) \equiv
    &
    (\textsf{step}(L_a,D)=(\theta,\sigma,\theta'))
    \Rightarrow 
  \\
    &~~~~~~~
    ((\theta \Rightarrow_P \sigma(Pre))
    \Rightarrow
    (\theta' \Rightarrow_P \sigma(Post)))
  \end{array}
  \]
  where $L_a$ is a normalized atom.
\end{definition}


\subsection{Status of Assertions and Partial Correctness}
\label{sec:asrt-status}

As mentioned before, the intended use of our assertions
is to perform debugging with respect to partial correctness, i.e.,
to ensure that the program does not produce unexpected results for
\emph{valid} (``expected'') queries.\footnote{In practice, this set
  of expected queries is determined from module interfaces that 
  define the set of exported predicates.}
Thus, we extend our notion of program to include assertions and valid
queries.

\begin{definition}[Annotated Program]
  An annotated program is a tuple $(P,\Q,\A)$ where $P$ is a
  (higher-order) \emph{constraint logic program} (as defined in
  Section~\ref{sec:Prel-Notat}), $\Q$ is a set of valid queries, and
  $\A$ is a set of assertions.  As before, $\AC$ denotes the set of
  \emph{calls} and \emph{success} assertion conditions derived from
  $\A$.
\end{definition}

In the context of annotated programs we extend $derivations$ to
operate on the set of valid queries as follows:
$\derivations{\Q}=\bigcup_{Q\in{\Q}} \derivations{Q}$.
We now provide several simple definitions which will be instrumental:

\begin{definition}[Assertion Condition Status]
  \label{def:assrts-cond-status}
  Given the set of queries \Q, the assertion condition $C$ can be 
  either \emph{checked} or \emph{false}, as follows:
  \begin{align*}
  \checkedAsr{C} &\equiv \forall D\in \derivations{\Q} ~.~ 
                                             \textsf{solve}(C,D) \\
  \falseAsr{C}   &\equiv \exists D\in \derivations{\Q} ~|~ \neg 
    \textsf{solve}(C,D)
  \end{align*}
\end{definition}
\begin{definition}[Assertion Status]
  \label{def:assrt-status}
  In an annotated program $(P,\Q,\A)$ an assertion $A \in \A$ is
  \emph{checked} (\emph{false}) if all (any) of the corresponding
  assertion conditions are \emph{checked} (\emph{false}).
\end{definition}
\begin{definition}[Partial Correctness]
  An annotated program $(P,\Q,\A)$ is \emph{partially correct} w.r.t.\
  the set of assertions $\A$ and the set of queries \Q~iff
   $\forall A \in \A$, $A$ is checked for \Q.
\end{definition}

\noindent
Note that it follows immediately that a program is partially correct
if all its assertion conditions are checked.
The goal of assertion checking is thus to determine whether each
assertion $A$ is false or checked for $\Q$.  
Again, for this it is sufficient to prove the corresponding assertions
conditions false or checked.
There are two kinds of approaches to doing this (which can also be
combined). 
While it is in general not possible to try all derivations stemming
from $\Q$, an alternative is to explore a hopefully representative set
of them~\cite{testchecks-iclp09}. Though this does not allow fully
validating the program in general, it makes it possible to detect many
incorrectness problems. This approach is explored in
Section~\ref{sec:Run-Time-Checking} in the context of our higher-order
derivations.
The second approach is to use global analysis techniques and is based
on computing safe approximations of the program behavior
statically~\cite{aadebug97-informal-short,ciaopp-sas03-journal-scp-short}. 
The extension of this approach to higher-order assertions is beyond
the scope of this paper.


\subsection{Operational Semantics for Higher-order Programs with
  First-order Assertions}
\label{sec:Run-Time-Checking}
We now provide an operational semantics which checks whether
assertion conditions hold or not while computing the (possibly
higher-order) derivations from a query.

\begin{definition}[Labeled Assertion Condition Instance]
  \label{def:assrt-cond-lab-inst}
  Given the atom $L_a$ and the set of assertion conditions $\AC$,
  $\LabAClit{L_a}$ denotes the set of \emph{labeled} assertion condition
  instances for $L_a$ of the form $\asrId{}\#C_a$, such that
  $\exists C \in \AC$,
  $C = \callsAsr{L}{Pre}$
  (or $C = \successAsr{L}{Pre}{Post}$),
  $\sigma$ is a renaming s.t. $L = \sigma(L_a)$,
  $C_a = \callsAsr{L_a}{\sigma(Pre)}$
  (or $C_a = \successAsr{L_a}{\sigma(Pre)}{\sigma(Post)}$),
  and $\asrId{}$ is an identifier that is unique for each $C_a$.
\end{definition}

In order to keep track of the violated assertion conditions, we
introduce an extended program state of the form
\exstate{G}{\theta}{\ADeps}, where $\ADeps$ denotes the set of
identifiers for falsified assertion condition instances.
We also extend the set of literals with syntactic objects of the form
$\checkLitLab{\asrId{}}$ where 
$\asrId{}$ is an identifier for an assertion condition instance,
which we call \emph{check literals}.
Thus, a \emph{literal} is now a constraint, an atom, a higher-order
literal, or a check literal.\footnote{While check literals are simply
  instrumental here, note that they are also directly useful for
  supporting program point assertions (which are basically check
  literals that appear in the body of
  rules)~\cite{assert-lang-disciplbook-short}. However, for simplicity
  we do not discuss program point assertions in this paper.}

\begin{definition}[Reductions in Higher-order Programs with First-order
  Assertions] 
  \label{def:ho-fo-reductions}
  A state $S = \exstate{L::G}{\theta}{\ADeps}$, where $L$ is a literal
  can be \emph{reduced} to a state $S'$, denoted \reductionA{S}{S'},
  as follows:
  \begin{enumerate}\setlength{\itemsep}{0mm}\setlength{\parskip}{0mm}
  \item 
    If $L$ is a constraint or $L = X(t_1,\ldots,t_n)$, then
    $S' = \exstate{G'}{\theta'}{\ADeps}$ where $G'$ and $\theta'$
    are obtained in a same manner as in
    \reduction{\state{L::G}{\theta}}{\state{G'}{\theta'}}
  \item 
    If $L$ is an atom and $\exists (L \mbox{\tt :-} B) \in \defn{L}$,
    then $S' = \exstate{B :: PostC :: G}{\theta}{\ADeps'}$ where:
    \[
      \ADeps' = \left\{
      \begin{array}{lr}
        \ADeps \cup \{\negAsrId{}\}
      & ~\text{if } \exists \; 
        \labCallsAsr{}{L}{Pre} \in \LabAClit{L} \text{ s.t. }
        \theta \not\Rightarrow_P Pre
      \\
        \ADeps 
      & \text{otherwise}
      \end{array}
      \right.
    \]
    and $PostC$ is the sequence 
    $\checkLitLab{\asrId{1}} :: \ldots :: \checkLitLab{\asrId{n}}$
    including all the checks $\checkLitLab{\asrId{i}}$ such that
    $\labSuccessAsr{i}{L}{Pre_i}{Post_i} \in \LabAClit{L} \;\land\; 
      \theta \Rightarrow_P Pre_i$.
  \item 
    If $L$ is a check literal $\checkLitLab{\asrId{}}$, then
    $S' = \exstate{G}{\theta}{\ADeps'}$ where:
    \[
      \ADeps' = \left\{
      \begin{array}{lr}
        \ADeps \cup \{\negAsrId{}\}
        & ~~\textrm{if}~ \labSuccessAsr{}{L}{\_}{Post} \in \LabAClit{L}
          \wedge \theta \not\Rightarrow_P Post 
        \\
        \ADeps 
        & \textrm{otherwise}
      \end{array}
      \right.
    \]
  \end{enumerate}
  Note that the order in which the \emph{PostC} \emph{check} literals
  are selected is irrelevant.
\end{definition}

The set of derivations for a program from its set of queries \Q\
using the semantics with assertions is denoted $\derivationsA{\Q}$.

\begin{definition}[Error-erased Derivation]
  \label{def:error-erased-derivation}
  The set of \emph{error-erased} derivations from $\reductionA{}{}$ is
  obtained by a syntactic rewriting $\errorErase{(-)}$ that removes
  states that begin by a check literal, check literals from goals, and
  the error set. It is recursively defined as follows:
  \begin{align*}
    \errorErase{\{D_1, \dots, D_n\}} &= \{\errorErase{D_1}, \dots,
    \errorErase{D_n}\} 
    \\
    \errorErase{(S_1, \dots, S_m, S_{m+1})} &= \left\{
      \begin{array}{lr}
        \errorErase{(S_1, \dots, S_m)} 
      & \text{if}~ S_{m+1} = \exstate{\checkLitLab{\_}::\_}{\_}{\_} 
      \\
        \errorErase{(S_1, \dots, S_m)} ~\Vert~
        (\errorErase{(S_{m+1})}) 
      & \text{otherwise}
      \end{array}
    \right.
    \\
    \errorErase{\exstate{G}{\theta}{\ADeps}} 
    &= \state{\errorErase{G}}{\theta} \\
    \errorErase{(L::G)} &= \left\{
      \begin{array}{lr}
        \errorErase{G} & \text{if}~ L = \checkLitLab{\_} \\
        L::(\errorErase{G}) & \text{otherwise}
      \end{array}
    \right. \\
    \errorErase{\emptyGoal} &= \emptyGoal
  \end{align*}
  \noindent where $\Vert$ stands for sequence concatenation.
\end{definition}

\begin{theorem}[Correctness and Completeness Under Assertion Checking]
  \label{th:ans-fail}
  For any annotated program $(P,\Q,\A)$, given ${\cal D} =
  \derivations{\Q}$ and ${\cal D}' = \derivationsA{\Q}$, it holds that
  ${\cal D}$ and ${\cal D}'$ are equivalent after filtering out check
  literals and error sets (formally defined as ${\cal D} =
  \errorErase{({\cal D}')}$ in Def.~\ref{def:error-erased-derivation}).
\end{theorem}

\begin{proof}
  We will prove ${\cal D} = \errorErase{({\cal D}')}$ by showing that 
  ${\cal D} \subseteq \errorErase{({\cal D}')}$
  and 
  ${\cal D} \supseteq \errorErase{({\cal D}')}$.
  \begin{itemize}
  \item ($\subseteq$) For all $D \in {\cal D}$ exists
      $D' \in {\cal D}'$ so that $D = \errorErase{(D')}$.
  \item ($\supseteq$) For all $D' \in {\cal D}'$, 
      $D = \errorErase{(D')} \in {\cal D}$.
  \end{itemize}
  We will prove each case:
  \begin{itemize}
  \item ($\subseteq$) Let $D = (S_1,\ldots,S_n)$, $S_i =
    \state{L_i}{\theta_i}$, for some $Q = (L_1,\theta_1) \in {\cal Q}$
    and $\reduction{S_i}{S_{i+1}}$. Proof by induction on the length
    $n$ of $D$:
    \begin{itemize}
    \item Base case ($n=1$). Let $S_1' =
      \exstate{L_1}{\theta_1}{\emptyset}$. It holds that
      $\errorErase{(S_1')} =
      \errorErase{\exstate{L_1}{\theta_1}{\emptyset})} =
      \state{\errorErase{L_1}}{\theta_1} = \state{L_1}{\theta_1} =
      S_1$ (since $L_1$ does not contain any check literal).
      Thus, $\errorErase{(D')} = \errorErase{((S_1'))} =
      (\errorErase{(S_1')}) = (S_1) = D$.
    \item Inductive case (show $n+1$ assuming $n$ holds).
      For each $D_2 = (S_1,\ldots,S_n,S_{n+1})$ there exists $D_2' =
      (S_1',\ldots,S_m',S_{m+1}')$ such that $\errorErase{(D_2')} =
      D_2$.
      Given the induction hypothesis
      it is enough to show that for each $\reduction{S_n}{S_{n+1}}$
      there exists $\reductionA{S_m'}{S_{m+1}'}$, such that
      $\errorErase{(S_{m+1}')} = S_{n+1}$.
      According to $\reductionA{}{}$ (see
      Def.~\ref{def:ho-fo-reductions}), $L_{m+1}'$ and $\theta_{m+1}'$
      are obtained in the same way than in $\reduction{}{}$ (see
      Def.~\ref{def:ho-reductions}), except for the introduction of
      check literals. Since all check literals are removed in
      error-erased states, it follows that $\errorErase{(S_{m+1}')} =
      S_{n+1}$.
      \hfill $\qed$
    \end{itemize}
  \item ($\supseteq$) Let $D' = (S_1',\ldots,S_m')$, $S_i' =
    \exstate{L_i'}{\theta_i'}{\ADeps_i}$, for some $Q =
    (L_1',\theta_1') \in {\cal Q}$ and
    $\reductionA{S_i'}{S_{i+1}'}$. Proof by induction on the length
    $m$ of $D'$:
    \begin{itemize}
    \item Base case ($m=1$). It holds that $\errorErase{(S_1')}= S_1$
      (showed in base case for $\subseteq$). Then $\errorErase{(D')} =
      D \in {\cal D}$.
    \item Inductive case (show $m+1$ assuming $m$ holds).
      We want to show that given $D_2' = (S_1',\ldots,S_m',S_{m+1}')$,
      $\errorErase{(D_2')} = D_2 \in {\cal D}$.
      Given the induction hypothesis
      it is enough to show that for each $\reductionA{S_m'}{S_{m+1}'}$
      there exists $\reduction{S_n}{S_{n+1}}$ such that $S_{n+1} =
      \errorErase{(S_{m+1}')}$ (so that $(S_1,\ldots,S_n,S_{n+1}) \in
      {\cal D}$) or $S_n = \errorErase{(S_{m+1}')}$ ($D_2 = D \in
      {\cal D}$).
      According to cases of Def.~\ref{def:ho-fo-reductions}:
      \begin{itemize}
      \item If $L_m'$ begins with a check literal then
        $\errorErase{(L_{m+1}')} = \errorErase{(L_m')}$. Thus
        $\errorErase{(S_{m+1}')} = \errorErase{(S_m')} = S_n$.
      \item Otherwise, it holds that $\errorErase{(S_{m+1}')}=
        S_{n+1}$ using the same reasoning than in the inductive case
        for $\subseteq$.
      \hfill $\qed$
      \end{itemize}
    \end{itemize}
  \end{itemize}
\end{proof}

This result implies that the semantics with assertions can also be
used to obtain all answers to the original query. Furthermore, the
following theorem guarantees that we can use the proposed operational
semantics for annotated programs in order to detect (all) violations
of assertions:

\begin{definition}[Run-time Valuations of an Assertion Condition on a
  Derivation]
  \label{def:rtsolve}
  Let $\ADeps(D)$ denote the error set of the last state
  of derivation 
  $D$, $\nthstate{D}{-1}=\exstate{\_}{\_}{\ADeps}$.
  The run-time valuation of an assertion condition $C$ on a derivation
  $D$ is given by:
  \[
    \textsf{rtsolve}(C,D) \equiv
    \forall \asrId{}, C', \sigma, L ~ (\asrId{}\#C' \in \LabAClit{L}
    \wedge \sigma(C) = C') \Rightarrow \ADeps(D) \nvdash \negAsrId{}
  \]
\end{definition}
I.e., condition $\textsf{rtsolve}(C,D)$ is valid if none of the possible
instances of the assertion condition $C$ are in the error set for
derivation $D$.

\begin{theorem}[Run-time Error Detection]
  \label{thm:rtcheck}
  For any annotated program $(P,\Q,\A)$, 
    $C \in \AC $ $\mbox{is false  iff } 
    \exists\; D \in \derivationsA{\Q} \; s.t.~ \neg\textsf{rtsolve}(C,D)$.
\end{theorem}

\begin{proof}
  $A \in \AC$ is false
\newline
  $\Leftrightarrow$
  from Def.~\ref{def:assrt-status} and Def.~\ref{def:assrt-cond}
  $\exists \{C_c,C_s\}$ assertion conditions s.t. 
  $\falseAsr{C_c} \vee \falseAsr{C_s}$, where
  $C_c = \callsAsr{L}{Pre}$ and
  $C_s = \successAsr{L}{Pre}{Post}$ correspond to $A$. Let us first
  prove $\neg\textsf{rtsolve}(C_c,D)$, and then 
  $\neg\textsf{rtsolve}(C_s,D)$. 

\hfill\newline
  $\falseAsr{C_c}$
\newline
  $\Leftrightarrow$
  from Def.~\ref{def:assrts-cond-status} 
  $\exists D \in \derivations{\Q}$ s.t. 
  $\neg\textsf{solve}(C_c,D)$
\newline
  $\Leftrightarrow$
  from Def.~\ref{def:assrt-valuation}
  $(\textsf{prestep}(L,D)=(\theta,\sigma)$ $\wedge$
  $\theta\not\Rightarrow_P\sigma(Pre))$
\newline
  $\Leftrightarrow$
  from Def.~\ref{def:ho-fo-reductions} 
  $\exists\; \reductionA{S}{S'}$ where:
\begin{align*}
  S 
  &= \exstate{L :: G}{\theta}{\ADeps} \text{ s.t. }
    \exists\; \labCallsAsr{}{L}{Pre} \in \LabAClit{L}
  \\
  S' 
  &= \exstate{\_}{\theta}{\ADeps'} \wedge 
    \ADeps' = \ADeps \cup \{\negAsrId{}\}
\end{align*}
  $\Leftrightarrow$
  from Def.~\ref{def:rtsolve} 
  $\neg\textsf{rtsolve}(C_c,D)$
\hfill $\qed$

\hfill\newline
  $\falseAsr{C_s}$
\newline
  $\Leftrightarrow$
  from Def.~\ref{def:assrts-cond-status} 
  $\exists D \in \derivations{\Q}$ s.t. 
  $\neg \textsf{solve}(C_s,D)$
\newline
  $\Leftrightarrow$
  from Def.~\ref{def:assrt-valuation}
  $(\textsf{step}(L,D)=(\theta,\sigma,\theta')$ $\wedge$
  $ \theta \Rightarrow_P \sigma(Pre)$ $\wedge$
  $ \theta' \not\Rightarrow_P \sigma(Post))$
\newline
  $\Leftrightarrow$
  from Def.~\ref{def:ho-fo-reductions} 
   $\exists\; 
  \reductionA{\reductionStarA{S}{S'}}{S''}$ where:
\begin{align*}
  S 
  &= \exstate{L::G}{\theta}{\_} \wedge
    \exists\; \labSuccessAsr{}{L}{Pre}{Post} \in \LabAClit{L}
    \wedge \theta \Rightarrow_P Pre
  \\
  S' 
  &= \exstate{\checkLitLab{\asrId{}}::G}{\theta'}{\ADeps'}
    \wedge \theta' \not\Rightarrow_P Post
  \\
  S'' 
  &= \exstate{\_}{\_}{\ADeps''} \wedge
    \ADeps'' = \ADeps' \cup \{\negAsrId{}\}
\end{align*}
  $\Leftrightarrow$
  from Def.~\ref{def:rtsolve} 
  $\neg\textsf{rtsolve}(C_s,D)$
\hfill $\qed$
\end{proof}

Th.~\ref{thm:rtcheck} states that assertion condition $C$ is false iff
there is a derivation $D$ in which the run-time valuation of the
assertion condition of $C$ in $D$ is false (i.e., if at least one
instance of the assertion condition $A$ is in the error set for such
derivation $D$).
Given a set of $false$ assertion conditions we can easily derive the set
of $false$ assertions using Def.~\ref{def:assrt-cond}. 
In order to prove that any assertion is checked this has to be
done for all possible derivations for all possible queries, which is
often not possible in practice. This is why analysis based on
abstractions is often used in practice for this purpose.


\section{Higher-order Assertions on Higher-order Derivations}
\label{sec:ho-ho}

Once we have established basic results for the case of first-order
assertions in the context of higher-order derivations, we extend the
notion of assertion itself to the higher-order case.  The motivation
is that in the higher-order context terms can be bound to predicates
and our aim is to also be able to state and check properties of such
predicates.


\subsection{Anonymous Assertions}

We start by generalizing the notion of assertion to include
\emph{anonymous assertions}: assertions where the predicate symbol is
a variable from $\VS$, which can be instantiated to any suitable
predicate symbol from $\PS$ to produce non-anonymous assertions. 
An anonymous assertion is an expression of the from 
``$\skbd{:- pred } L \skbd{ : } Pre \skbd{ => }Post$'', where $L$
is of the form $X(V_1, \ldots,V_n)$ and $Pre$ and $Post$ are
DNF formulas of \emph{prop} literals.

\begin{example}
  The anonymous assertion 
  ``\skbd{:- pred X(A,B) : list(A) => list(B).}''
  states that any predicate $p \in P$ that $X$ is constrained to
  should be of arity 2, it should be called with its first argument
  instantiated to a list, and if it succeeds, then its second argument
  should be also a list on success.
\end{example}

We now introduce \emph{predprops}, which 
gather a number of anonymous assertions in order to fully describe
variables containing higher-order terms (predicate symbols), similarly
to how \emph{prop} literals describe conditions for variables
containing first-order terms. 

\begin{definition}[Predprop]
  \label{def:pp-literal}
  Given $Pre_i$ and $Post_i$ conjunctions of \emph{prop} literals, a
  \emph{predprop} $pp(X)$ is an expression of the form:
  \[
  \begin{array}{ll}
    \kbd{pp(X)}  \kbd{\{}    
  & 
    \kbd{:- pred } X(V_1,\ldots,V_m)
    \kbd{ : } Pre_1 \kbd{ => } Post_1.
  \\ 
  & 
       \ldots
  \\
  & 
    \kbd{:- pred } X(V_1,\ldots,V_m)
    \kbd{ : } Pre_n \kbd{ => } Post_n. ~~\kbd{\}}
  \end{array}
  \]
\end{definition}

\begin{definition}[Anonymous Assertion Conditions for a 
  \emph{predprop}]
  \label{def:assrts-mapping}
  The corresponding set of \emph{anonymous assertion conditions}
  for the predprop $pp(X)$ is defined as 
  $\ACpp{pp(X)} = \{C_i[X]\;|\;i=0..n\}$ where:
  \[
    C_i[X] = \left\{
    \begin{array}{ll}
      \callsAsr{X(V_1,\ldots,V_m)}{Pre} 
    & ~~~~i = 0 
    \\
      \successAsr{X(V_1,\ldots,V_m)}{Pre_i}{Post_i}   
    & ~~~~i = 1..n
    \end{array}
    \right.
  \]
  The variable $X$ can be instantiated to a particular predicate
  symbol $q \in \PS$ to produce a set of non-anonymous assertion
  conditions $\ACpp{pp(p)}$ for $q$.
\end{definition}

\begin{example}
  Consider defining a \skbd{comparator(Cmp)} \emph{predprop} that
  describes predicates of arity 3 which can be used to compare
  numerical values:
  \begin{small}
    \begin{verbatim}
      comparator(Cmp) {
        :- pred Cmp(X,Y,Res) : int(X),int(Y) => between(-1,1,Res).
        :- pred Cmp(X,Y,Res) : flt(X),flt(Y) => between(-1,1,Res). }.
    \end{verbatim}
  \end{small}
  The \kbd{comparator(Cmp)} \emph{predprop} includes two anonymous
  assertions describing a set of possible preconditions and
  postconditions for predicates of this kind. In this example:
  \begin{small}
    \[
    \begin{array}{ll}
      \ACpp{comparator(Cmp)} = \{ 
    & \\
    ~~\callsAsr{Cmp(X,Y,Res)}{(int(X) \wedge int(Y)) \vee 
                                           (flt(X) \wedge flt(Y))},
    & \\
    ~~\successAsr{Cmp(X,Y,Res)}{int(X) \wedge int(Y)}{between(-1,1,Res)}
    & \\
    ~~\successAsr{Cmp(X,Y,Res)}{flt(X) \wedge flt(Y)}{between(-1,1,Res)} 
    & 
      \}
    \end{array}
    \]
  \end{small}
\end{example}

\begin{example}
  Fig.~\ref{ex:pp-calls} provides a larger example. This example
  is more stylized for brevity, but it covers a good subset of the
  relevant cases, 
  used later to illustrate the semantics.
\end{example}

\begin{definition}[Meaning of a \emph{predprop} Literal]
  \label{def:pp-meaning}
  The \emph{meaning} of a \emph{predprop} $pp(X)$, denoted $|pp(X)|$
  is the set of constraints 
  $\{ X = q ~|~ q \in \PS,\; \forall \_\#C \in \ACpp{pp(q)}:
    \checkedAsr{C} \}$.
\end{definition}

A predicate given by its predicate symbol $p \in \PS$ is
\emph{compatible}
with a \emph{predprop} $pp(X)$ if all the assertions resulting from
$pp(p)$ are \emph{checked} for all possible queries in an annotated
program.


\begin{figure}
\vspace*{-1mm}
  \centering
    \begin{small}
\begin{verbatim}
        :- nneg(P) { :- pred P(X) => nnegint(X). }.
        :- neg(P)  { :- pred P(X) =>  negint(X). }.

        :- pred test_c(P,N) : nneg(P).
        :- pred test_c(P,N) :  neg(P).
        test_c(P,N) :- P(N).

        :- pred test_s(N,P) : nnegint(N) => nneg(P).
        :- pred test_s(N,P) :  negint(N) =>  neg(P).
        test_s( 1,P) :- P = z. % bug here, should be P = p
        test_s(-1,P) :- P = n.

        z(1). z(-2).   p(1). p(2).   n(-1). n(-2).   c(a). c(b).
\end{verbatim}
    \end{small}
  \caption{Sample Program with predprops.}
  \label{ex:pp-calls}
\vspace*{-1mm}
\vspace*{-4mm}
\end{figure}


\subsection{Operational Semantics for  Higher-order Programs with
Higher-order Assertions}
We now discuss several alternative operational semantics for
higher-order programs with higher-order assertions. In all cases the
aim of the semantics is to check whether assertions with
\emph{predprops} hold or not during the computation of the
derivations from a query.


\subsubsection{Checking with Static \emph{predprops}} 
According to Definition~\ref{def:pp-meaning}, a \emph{predprop}
literal $pp(X)$ denotes the subset of predicates for which
all the associated assertions are checked.
When that set of assertions can be statically computed, then $\theta
\Rightarrow_P Cond$ can be used for both prop and predprop $Cond$
literals, and the operational semantics is identical
to the one for the higher-order programs and regular assertions.

We will denote as $\reductionHAs{S}{S'}$ a reduction from a state $S$
to a state $S'$ under the semantics for higher-order derivations in
programs with assertions that may contain higher-order properties,
which are statically precomputed.
Thus, state reductions are performed as follows:
\[
\frac{
  \reductionA{\exstate{G}{\theta}{\ADeps}}
             {\exstate{G'}{\theta'}{\ADeps'}}
}{\reductionHAs{\exstate{G}{\theta}{\ADeps}}
              {\exstate{G'}{\theta'}{\ADeps'}}}
\]
\noindent The meaning of each predprop, $|pp(X)|$, can be inferred or
checked (if given by the user) by static analysis.

In this semantics, given the program shown in Fig.~\ref{ex:pp-calls}
and the goal \texttt{test\_c(z,-2)}, assertions are detected to be
false since $\{P=\texttt{z}\} \not\subset |\texttt{neg(P)}|$ and 
$\{P=\texttt{z}\} \not\subset |\texttt{nneg(P)}|$.


\subsubsection{Checking with Dynamic \emph{predprops}}
Given the difficulty in determining the meaning of $|pp(X)|$
statically, we also propose a semantics with dynamic checking.
We start with an over-approximation of each predprop $|pp(X)| = \{ X =
p ~|~ p \in \PS \}$ and incrementally remove predicate symbols, as
violations of assertion conditions are detected:
\begin{itemize}
\item we can detect when some assertion condition instance is violated
  (Def.~\ref{def:ho-fo-reductions});
\item we need a way to obtain a set of assertion condition instances
  from predprops (anonymous asserion condition instances);
\end{itemize}
We do that by defining instantiations of anonymous assertion
conditions for particular predicate symbols and the dependencies among
those instances.

The following two definitions extend the notion of assertion condition
instances from Def.~\ref{def:assrt-cond-lab-inst} to the case of
anonymous assertion conditions and higher-order literals:
\begin{definition}[Labeled Hypothetical Assertion Condition]
  \label{def:lab-hyp-assrt-cond}
  Given a predprop $pp(X)$ and a predicate symbol $p \in \PS$,
  $\LabACpp{pp(p)}$ denotes the set of \emph{labeled} hypothetical
  assertion conditions of the form $\hypId{}\#C_p$,
  such that
  $C[X] \in \ACpp{pp(X)}$ (Def.~\ref{def:assrts-mapping}),
  $L = X(V_1,\ldots,V_n)$,
  $L_p = p(V_1,\ldots,V_n)$,
  $C_p$ is defined as:
  \[
    C_p = \left\{
    \begin{array}{ll}
      \callsAsr{L_p}{Pre} 
    & ~~\textrm{if}~~ C[X] = \callsAsr{L}{Pre}
    \\
      \successAsr{L_p}{Pre}{Post}   
    & ~~\textrm{if}~~ C[X] = \successAsr{L}{Pre}{Post}
    \end{array}
    \right.
  \]
  and $\hypId{}$ is an identifier that is unique for each $C_p$.
\end{definition}

In this semantics we allow the assertion condition instances to be
derived from the hypothetical assertion conditions in the same way, as
in Def.~\ref{def:assrt-cond-lab-inst}.
However, the violation of such an instance has to be treated in a
special way, as it does not signal the violation of its conditions,
but instead of the corresponding predprop.
For simplicity, we also introduce a special label \hypId{0} to denote
the assertion conditions that appeared originally in the program.
The error set $\ADeps$ in Def.~\ref{def:ho-fo-reductions} contained
negated assertion condition instance identifiers. Now we extend this
set with \emph{assertion dependency rules} of the form 
$\bigwedge (\bigvee \negAsrId{}) \rightarrow \negAsrId{}$.
The following definitions provide the description of how such
dependencies are generated.

\begin{definition}[Literal Simplification]
  The \emph{simplification} of a literal $L$ w.r.t.
  $\theta$ is defined as:
  \[
    \simp{L} = \left\{
    \begin{array}{ll}
      L    
    & ~~\textrm{if}~~ L \textrm{ is a predprop}
    \\
      true   
    & ~~\textrm{if}~~ \theta \Rightarrow_P L 
    \\
      false  
    & ~~\textrm{if}~~ \theta \not\Rightarrow_P L
    \end{array}
    \right.
  \]
  We extend this definition for a conjunction of
  literals.
\end{definition}

\begin{definition}[Extension of $\AC$ and $\ADeps$ for dynamic
  predprop checking]
  \label{def:ae-extension}
  Given the label $\asrId{}$ of an assertion condition instance and a
  formula of the form $Props =
  \bigvee_{i=1}^{n}(\bigwedge_{j=0}^{m(i)}Prop_{ij})$, where
  $Prop_{ij}$ is either a prop or predprop literal, the
  \emph{extension} of $\AC$ and $\ADeps$ for dynamic
  predprop checking, denoted as
  $\ext{\asrId{}}{Props} = 
  (\Delta\AC,\Delta\ADeps)$, is obtained as follows:
  \begin{enumerate}
  \item
    if $\simp{Props} = false$,
    then $\Delta\AC = \emptyset$ and $\Delta\ADeps = \{\negAsrId{}\}$;
  \item
    otherwise: $\Delta\AC = \bigcup_{i=1}^n \AC^i,\;$ and 
    $ \Delta\ADeps = \{\bigwedge_{i = 1}^{n}
      (\bigvee_{\hypId{} \in H_i} \negHypId{})
      \rightarrow \negAsrId{}\}$ where:
    \[
    \begin{array}{rl}
      \AC^i = 
      &\{\hypId{}\#C \in \LabACpp{Prop_{ij}}\;|\;
      0 \leq j \leq m(i), Prop_{ij} = pp_{ij}(X_{ij}),
    \\
      &\;\;\;\; pp_{ij}(X_{ij}) \textrm{ is a predprop and } X_{ij} 
      \textrm{ is bound to some } q \in \PS\}.
    \\
      H_i =
      & \{\hypId{}\;|\;\hypId{}\#\_ \in \AC^i\},\;
    \end{array}
    \]
  \end{enumerate}
\end{definition}


We will denote as $\reductionHAd{S}{S'}$ a reduction from a state $S$
to a state $S'$ under the current semantics. 
\begin{definition}[Reductions in Higher-order Programs with 
Higher-order Assertions]
  \label{def:ho-ho-reductions}
  A state $S = \exstate{L::G}{\theta}{\ADeps}$, where $L$ is a literal
  can be \emph{reduced} to a state $S'$, denoted \reductionHAd{S}{S'},
  as follows:
  \begin{enumerate}\setlength{\itemsep}{0mm}\setlength{\parskip}{0mm}
  \item 
    If $L$ is a constraint or $L = X(t_1,\ldots,t_n)$, then
    $S' = \exstate{G'}{\theta'}{\ADeps}$ where $G'$ and $\theta'$
    are obtained in a same manner as in
    \reductionA{\state{L::G}{\theta}}{\state{G'}{\theta'}};
  \item 
    If $L$ is an atom and $\exists (L \mbox{\tt :-} B) \in \defn{L}$,
    then for each $\asrId{i}\#C_i \in \LabAClit{L}$:
      \begin{align*}
      \hypId{i} 
      &= \left\{
        \begin{array}{ll}
        \hypId{}
        & \text{if } C_i\text{ is an instance of some }\hypId{}\#C 
          \in \AC 
        \\
        \hypId{0}
        & otherwise
        \end{array} \right.
      \\
      (\Delta_i\AC,\Delta_i\ADeps) 
      &= \left\{
        \begin{array}{ll}
          \ext{\asrId{i}}{Pre} 
          & \textrm{if } C_i = \callsAsr{L}{Pre}
          \\
          (\emptyset,\emptyset) 
          & \textrm{otherwise}
        \end{array} \right.
      \\
      PostC_i 
      &=\left\{
        \begin{array}{ll}
          \checkLitLab{\asrId{i}} 
        & \text{if } C_i = \successAsr{L}{Pre_i}{Post_i} 
        \\
        & \text{ and } \simp{Pre_i} = true
        \\
          true
        & \textrm{otherwise}
        \end{array} \right.
    \end{align*}
      and $S' = \exstate{B :: PostC :: G}{\theta}{\ADeps'}$,
      where $\ADeps' = 
       \ADeps \cup \bigcup_i \{\negAsrId{i} \rightarrow \negHypId{i}\}
       \cup \bigcup_i \Delta_i\ADeps$,
      $\AC' = \AC \cup \bigcup_i \Delta_i\AC$,
      and $PostC$ is the sequence $PostC_1::\ldots::PostC_n$
      (simplifying \texttt{true} literals).
  \item 
    If $L$ is a check literal $\checkLitLab{\asrId{}}$ and
    $\labSuccessAsr{}{L'}{\_}{Post} \in
    \LabAClit{L'}$, then $S' = \exstate{G}{\theta}{\ADeps'}$ where
    $(\Delta\ADeps,\Delta\AC) = \ext{\asrId{}}{Post}$,
    $\ADeps' = \ADeps \cup \Delta\ADeps$ and $\AC' = \AC \cup
      \Delta\AC$.
  \end{enumerate}
\end{definition}

Note that in this semantics we support more than one \emph{calls}
assertion condition per predicate (as several predprops may be applied
to the same predicate symbol).
Also note that in general we cannot prove with dynamic checking that
a predprop is $true$. So, as a safe approximation we treat
preconditions in such \emph{success} assertion conditions as
$false$.

\begin{definition}[Trivial Assertion Condition]
  An assertion condition $C$ is \emph{trivial} if it is of the form
  $\callsAsr{\_}{true}$ or $\successAsr{\_}{\_}{true}$. It is also
  assumed that for any predprop $pp(X)$ $\ACpp{pp(X)}$ does not contain
  trivial assertion conditions.
\end{definition}

\begin{theorem}[Higher-order Run-time Checking]
  \label{thm:rtcheckHO}
  For any annotated program $(P,\Q,\A)$, if
  $\exists D \in \derivationsHAd{\Q} \textsf{ s.t. } 
    \neg\textsf{rtsolve}(C,D) \Rightarrow$ 
  $C \in \AC \textrm{ is } false$.
\end{theorem}

\begin{proof}
  In this proof we reflect the case when an assertion condition is
  falsified because of some of its predprops violation.
  To do so it is enough to show that at most one predprop was violated.  
  Let us first prove the theorem for the case when the falsified
  assertion condition is $C_c = \callsAsr{L}{pp(X)}$ and then for the
  case $C_s = \successAsr{L}{Pre}{pp(X)}$, where $pp(X)$ is a predprop.
  Without the loss of generality we assume that $\ACpp{pp(X)}$ has
  cardinality of 1 (which is a case when $pp(X)$ consists of one
  anonymous assertion and one of the corresponding anonymous assertion
  conditions is trivial).

\hfill\newline  
  $\neg\rtsolve{C_c}{D}$
\newline
  $\Leftrightarrow$ From Def.~\ref{def:rtsolve}: 
  $\exists \asrId{}', C_c', \sigma, L ~ (\asrId{}'\#C_c' \in \LabAClit{L})
    \wedge (\sigma(C_c) = C_c') \wedge (\ADeps(D) \vdash \negAsrId{}')$
\newline
  $\Rightarrow$ From Def.~\ref{def:ho-ho-reductions} and $\ADeps(D)
   \vdash \negAsrId{}'$ it must hold 
  $D = (\ldots,S_1,\ldots,S_2,S_3\ldots,S_4,\ldots)$ where:
  \[
  \begin{array}{ll}
    S_1 = \exstate{L'::\_}{\theta_1}{\_}
    & \text{s.t. } \exists\; \clause{L'}{B'} \in \defn{L}, 
      \asrId{}'\#\callsAsr{L'}{\sigma(pp(X))} \in \LabAClit{L},
    \\
    & ~~~~~\theta_1 \models (X = q), q \in \PS
    \\
    S_2 = \exstate{L_2::\_}{\_}{\ADeps_2}
    & \text{s.t. }
      \{\negHypId{}   \rightarrow \negAsrId{}',
        \negAsrId{}' \rightarrow \negHypId{0} ,\} \in \ADeps_2,
      \hypId{}\#C_q \in \LabACpp{pp(q)}, L_2 = q(\ldots)
    \\
    S_3 = \exstate{\_}{\_}{\ADeps_3}
    & \text{s.t. } 
      \{\negAsrId{}''  \rightarrow \negHypId{} \} \in \ADeps_3,\;
      \asrId{}''\#C_c'' \in \LabAClit{L_2}
    \\
    S_4 = \exstate{\_}{\_}{\ADeps_4}
    & \text{s.t. } \ADeps_4 \vdash \negAsrId{}''
  \end{array}
  \]
  $\Rightarrow$ From $\ADeps_3 \vdash \negAsrId{}''$ and
  Th.~\ref{thm:rtcheck} we know that $\neg\checkedAsr{C_c''}$
  and thus $(X = q) \not\in |pp(X)|$ according to
  Def.~\ref{def:pp-meaning}.
\newline
  $\Rightarrow$ From Def.~\ref{def:succeeds-triv}
  it follows that $\theta_3 \not\Rightarrow_P pp(q)$
\newline
  $\Rightarrow$
  Given the state $S_1$ before the call to $L'$ and the
  state $S_3$:
  $ ( \textsf{prestep}(L,D)=(\theta_3,\sigma) ) $
  $ \wedge $
  $ (\theta' \not\Rightarrow_P\sigma(pp(X))) $
\newline
  $\Rightarrow$ From Def.~\ref{def:assrt-valuation}
  $\neg\textsf{solve}(C_c,D)$
  $\Rightarrow$ From Def.~\ref{def:assrts-cond-status} $\falseAsr{C_c}$
  \hfill $\qed$

\hfill\newline  
  $\neg\rtsolve{C_s}{D}$
\newline
  $\Leftrightarrow$ From Def.~\ref{def:rtsolve}: 
  $\exists \asrId{}', C_s', \sigma, L ~ (\asrId{}'\#C_s' \in \LabAClit{L})
    \wedge (\sigma(C_s) = C_s') \wedge (\ADeps(D) \vdash \negAsrId{}')$
\newline
  $\Rightarrow$ From Def.~\ref{def:ho-ho-reductions} and $\ADeps(D)
  \vdash \negAsrId{}'$ it must hold 
  $D =
  (\ldots,S_1,S_2,\ldots,S_3,S_4,\ldots,S_5,S_6,\ldots,S_7,\ldots)$
  where: 
  \[
  \begin{array}{ll}
    S_1 = \exstate{L'::\_}{\theta_1}{\_} & \text{s.t.}~
      \exists\; \clause{L'}{B'} \in \defn{L}, \\
    & ~~~~~\asrId{}'\#\successAsr{L'}{\sigma(Pre)}{\sigma(pp(X))} \in
      \LabAClit{L}, 
    \\
    & ~~~~~\theta_1 \Rightarrow_P \sigma(Pre) 
    \\
    S_2 = \exstate{B'::\checkLitLab{\asrId{}'}::\_}{\_}{\ADeps_2} 
    & \text{s.t. }
      \{\negAsrId{}' \rightarrow \negHypId{0}\} \in \ADeps_2
    \\
    S_3 = \exstate{\checkLitLab{\asrId{}'}::\_}{\_}{\_}
    &
    \\
    S_4 = \exstate{\_}{\theta_4}{\ADeps_4} 
    & \text{s.t. }\theta_4 \models (X = q), q \in \PS, 
     \{\negHypId{} \rightarrow \negAsrId{}'\} \in \ADeps_4,
    \\
    & ~~~~~\hypId{}\#C_q \in \LabACpp{pp(q)}.
    \\
    S_5 = \exstate{L_5::\_}{\_}{\_}
    & \text{s.t. } L_5 = q(\ldots)
    \\
    S_{6} = \exstate{\_}{\_}{\ADeps_6}
    & \text{s.t. } \{\negAsrId{}'' \rightarrow \negHypId{}\} \in
     \ADeps_6 \text{ where } \asrId{}''\#C_s'' \in \LabAClit{L_5}
    \\
    S_7 = \exstate{\_}{\theta_7}{\ADeps_7}
    & \text{s.t. } \ADeps_7 \vdash \negAsrId{}''
  \end{array}
  \] 
  $\Rightarrow$ From $\ADeps_7 \vdash \negAsrId{}''$ and
  Th.~\ref{thm:rtcheck} we know that $\neg\checkedAsr{C_s''}$
  and thus $(X = q) \not\in |pp(X)|$ according to
  Def.~\ref{def:pp-meaning}.
\newline
  $\Rightarrow$ From Def.~\ref{def:succeeds-triv} it follows that 
  $\theta_7 \not\Rightarrow_P pp(q)$
\newline
  $\Rightarrow$ Given the state $S_1$ before the call to $L'$ and the
  state $S_7$: 
  $(\textsf{step}(L,D)=(\theta_1,\sigma,\theta_7))$ $\wedge$
  $(\theta_1 \Rightarrow_P \sigma(Pre))$ $\wedge$
  $(\theta_7 \not\Rightarrow_P \sigma(pp(X)))$ for
  $\asrId{}'\#C_s' \in \LabAClit{L}$
\newline
  $\Rightarrow$ From Def.~\ref{def:assrt-valuation}
  $\neg\textsf{solve}(C_s,D)$
  $\Rightarrow$ From Def.~\ref{def:assrts-cond-status} $\falseAsr{C_s}$
  \hfill $\qed$

\end{proof}

Let us trace finished derivations $D^1,D^2$ and $D^3$ from the queries
$Q_1 = (\kbd{test\_c(n,X)}, true)$, $Q_2 = (\kbd{test\_c(c,X)}, true)$
and $Q_3 = (\kbd{(test\_s(1,P),P(-2))}, true)$, respectively, to the
program in Fig.~\ref{ex:pp-calls}.

\begin{table}[h]
  \caption{A derivation of the query (\kbd{test\_c(n,X)}, $true$) to
    the program in Fig.~\ref{ex:pp-calls}.}
  \begin{minipage}{\textwidth}
  \small
    \begin{tabular}{p{0.18\textwidth}|p{0.14\textwidth}%
                    p{0.18\textwidth}p{0.46\textwidth}}
      \hline
      G         
          & $\Delta\theta$  
          & $\Delta\ADeps$ 
          & $\Delta$(labeled instances + hypothetic $\AC$)
      \\ \hline
      {\tt test\_c(n,X)} 
          & $P = n$\newline
            $N = -1$\newline
            $X = N$    
          & $\negAsrId{1} \rightarrow \negHypId{0}$\newline
            $\negHypId{1} \wedge \negHypId{2}
             \rightarrow \negAsrId{1}$
          & $\labCallsAsr{1}{test\_c(n,X)}{nneg(n) \vee neg(n)}$\newline
            $\hypSuccessAsr{1}{n(Z)}{true}{nnegint(Z)}$\newline
            $\hypSuccessAsr{2}{n(Z)}{true}{negint(Z)}$ 
      \\ \hline
      {\tt P(-1)}
          & $Z = -1$         
          & $\negAsrId{2} \rightarrow \negHypId{1}$\newline
            $\negAsrId{3} \rightarrow \negHypId{2}$
          & $\labSuccessAsr{2}{n(-1)}{true}{nnegint(-1)}$\newline
            $\labSuccessAsr{3}{n(-1)}{true}{negint(-1)}$ 
      \\ \hline
      \checkLitLab{$\asrId{2}$},\newline
      \checkLitLab{$\asrId{3}$}         
          & -         
          & $\negAsrId{2}$
          & -
      \\ \hline
      \checkLitLab{$\asrId{3}$} 
          & -          
          & -
          & -
      \\ \hline
      \emptyGoal
          & -
          & -
          & -
      \\ \hline
    \end{tabular}
  \ \\
  \end{minipage}
  \label{tbl:test1}
\end{table}

In  $\nthstate{D^1}{1}$ we encounter two assertions for
\kbd{test\_c/2} with a predprop in each precondition and trivial
postconditions.
According to state reduction rules, $\Delta\AC$ consists of calls
assertion condition instance $\asrId{1}$ and two hypothetical
assertion conditions $\hypId{1}$ and $\hypId{2}$, derived from
predprops \skbd{nneg/1} and \skbd{neg/1}, and $\Delta\ADeps =
\{\negAsrId{1} \rightarrow \negHypId{0}, \negHypId{1} \wedge
\negHypId{2} \rightarrow \negAsrId{1}\}$.
In $\nthstate{D^1}{2}$ and current goal \kbd{P(-1)} (which is
implicitly reduced as \kbd{n(-1)}), success assertion condition
instances $\asrId{2}$ and $\asrId{3}$ are derived from the hypotheses
$\hypId{1}$ and $\hypId{2}$, and $\Delta\ADeps = \{\negAsrId{2}
\rightarrow \negHypId{1}, \negAsrId{3} \rightarrow \negHypId{2} \}$.
Consequently, two check literals, $\checkLitLab{2}$ and
$\checkLitLab{3}$ are added to the goal sequence.
In states $\nthstate{D^{1}}{3}$ and $\nthstate{D^1}{4}$ those literals
are reduced, which results in adding $\negAsrId{2}$ to $\ADeps$
because \skbd{nnegint(-1)} property from the postcondition of
$\asrId{2}$ is violated.
This example shows that the mechanism of dependencies between
assertion conditions allows to avoid ``false negative'' results in
assertion checking.

\begin{table}[h]
  \caption{A derivation of the query (\kbd{test\_c(c,X)}, $true$) to
    the program in Fig.~\ref{ex:pp-calls}.}
  \begin{minipage}{\textwidth}
    \small
    \begin{tabular}{p{0.18\textwidth}|p{0.14\textwidth}%
                    p{0.18\textwidth}p{0.46\textwidth}}
      \hline
      G         
          & $\Delta\theta$  
          & $\Delta\ADeps$ 
          & $\Delta$(labeled instances + hypothetic $\AC$)
      \\ \hline
      {\tt test\_c(c,X)}     
          & $P = c$\newline
            $N = a$\newline
            $X = N$    
          & $\negAsrId{1} \rightarrow \negHypId{0}$\newline
            $\negHypId{2} \wedge \negHypId{3}
             \rightarrow \negAsrId{1}$
          & $\labCallsAsr{1}{test\_c(c,X)}{nneg(c) \vee neg(c)}$\newline
            $\hypSuccessAsr{2}{c(Z)}{true}{nnegint(Z)}$\newline
            $\hypSuccessAsr{3}{c(Z)}{true}{negint(Z)}$
      \\ \hline
      {\tt P(a)} 
          & $Z = a$         
          & $\negAsrId{2} \rightarrow \negHypId{2}$\newline
            $\negAsrId{3} \rightarrow \negHypId{3}$
          & $\labSuccessAsr{2}{c(a)}{true}{nnegint(a)}$\newline
            $\labSuccessAsr{3}{c(a)}{true}{negint(a)}$ 
      \\ \hline
      \checkLitLab{$\asrId{2}$},\newline
      \checkLitLab{$\asrId{3}$}        
          & -       
          & $\negAsrId{2}$
          & -
      \\ \hline
      \checkLitLab{$\asrId{3}$} 
          & -          
          & $\negAsrId{3}$
          & -
      \\ \hline 
      \emptyGoal
         & -
         & -
         & -
      \\ \hline
    \end{tabular}
  \end{minipage}
  \label{tbl:test2}  
\end{table}

The derivation $D^2$ is similar to $D^1$. The difference is in
$\nthstate{D^2}{4}$ state, when it becomes possible to infer $\ADeps
\vdash \negAsrId{1}$ and thus to conclude that $\kbd{c/1} \not\in
|nneg(X)| \wedge \kbd{c/1} \not\in |neg(X)|$ and that both assertions
for \kbd{test\_c/2} are $false$ for this query.

\begin{table}[h]
  \caption{A finished derivation of the query 
    (\kbd{(test\_s(1,P),P(-2))}, $true$) to the program in 
    Fig.~\ref{ex:pp-calls}.}
  \begin{minipage}{\textwidth}
    \small
    \begin{tabular}{p{0.17\textwidth}|p{0.13\textwidth}%
                    p{0.15\textwidth}p{0.52\textwidth}}
      \hline
      G         
         & $\Delta\theta$  
         & $\Delta\ADeps$ 
          & $\Delta$(labeled instances + hypothetic $\AC$)
      \\ \hline
      {\tt test\_s(1,P),}\newline
      \kbd{P(-2)}
          & $N = 1$        
          & $\negAsrId{0} \rightarrow \negHypId{0}$ \newline
            $\negAsrId{1} \rightarrow \negHypId{0}$
          & $\labCallsAsr{0}{test\_s(1,P)}{nnegint(1) \vee negint(1)}$
            \newline
            $\labSuccessAsr{1}{test\_s(1,P)}{nnegint(1)}{nneg(P)}$     
      \\ \hline
      {\tt P = z, check(\asrId{1}),}\newline
      \texttt{P(-2)}
          & $P = z$
          & -
          & -
      \\ \hline
      {\tt check(\asrId{1}), P(-2)}
          & -
          & $\negHypId{2} \rightarrow \negAsrId{1}$
          & $\hypSuccessAsr{2}{z(Z)}{true}{nnegint(Z)}$
      \\ \hline
      \texttt{P(-2)}
          & $Z = -2$
          & $\negAsrId{2} \rightarrow \negHypId{2}$
          & $\labSuccessAsr{2}{z(-2)}{true}{nnegint(-2)}$
      \\ \hline
      \checkLitLab{$\asrId{2}$}
        & -
        & $\negAsrId{2}$ 
        & -
      \\ \hline
      \emptyGoal
          & -        
          & -
          & -
      \\ \hline
    \end{tabular}
  \end{minipage}
  \label{tbl:test5}
\end{table}

In $\nthstate{D^3}{1}$ we encounter two assertions with a predprop in
each postcondition.
According to state reduction rules, $\Delta\AC$ for this state
consists of calls and success assertion condition instances,
$\asrId{1}$ and $\asrId{2}$, $\Delta\ADeps = \{ \negAsrId{0}
\rightarrow \negHypId{0}, \negAsrId{1} \rightarrow \negHypId{0} \}$
for them.
Also, a check literal $\checkLitLab{\asrId{1}}$ is added to the goal
sequence.
After its reduction a hypothetical assertion condition $\hypId{2}$,
derived from \kbd{nneg(X)} predprop, is added to $\AC$ in
$\nthstate{D^3}{3}$, and $\ADeps$ is extended with a dependency rule
$\{\negHypId{2} \rightarrow \negAsrId{1}\}$.
In state $\nthstate{D^3}{4}$ an assertion condition instance
$\asrId{2}$ is obtained from $\hypId{2}$ and $\Delta\ADeps = \{
\negAsrId{2} \rightarrow \negHypId{2}\}$.
Finally, in state $\nthstate{D^3}{5}$ it becomes possible to infer
$\ADeps \vdash \negAsrId{1}$ and thus detect that the corresponding
assertion for \kbd{test\_s/2} predicate is $false$ because of the
predprop \kbd{nneg(X)} violation.


\section{Conclusions and Future Work}

This paper contributes towards filling the gap between higher-order
(C)LP programs and assertion-based extensions for error detection and 
program validation.
To this end we have defined a new class of properties, ``predicate
properties'' (\emph{predprops} in short), and proposed a syntax and
semantics for them. These new properties can be used in assertions for
higher-order predicates to describe the properties of the higher-order
arguments.
We have also discussed several operational semantics for performing
run-time checking of programs including predprops and provided
correctness results. 

Our predprop properties specify conditions for predicates that are
independent of the usage context. This corresponds in functional
programming to the notion of \emph{tight} contract
satisfaction~\cite{DBLP:journals/toplas/DimoulasF11}, and it contrasts
with alternative approaches such as \emph{loose} contract
satisfaction~\cite{DBLP:conf/icfp/FindlerF02}.
In the latter, contracts are attached to higher-order arguments by
implicit function wrappers. The scope of checking is local to the
function evaluation. Although this is a reasonable and pragmatic
solution, we believe that our approach is more general and more
amenable for combination with static verification techniques.
For example, avoiding wrappers allows us to remove checks (e.g., by
static analysis) without altering the program semantics.
\footnote{E.g. \texttt{f(g)=g} is not an identity function if wrappers
  are added to \texttt{g} on call. This complicates reasoning about the
  program, and may lead to unexpected and hard to detect differences in
  program semantics. Similar examples can be constructed where the
  presence of predprops in assertions would invalidate many reasonable
  program transformations.}
Moreover, our approach can easily support \emph{loose} contract
satisfaction, since it is straightforward in our framework to
optionally include wrappers as special predprops.

We have included the proposed predprop extensions in an experimental
branch of the Ciao assertion language implementation. This has the
immediate advantage, in addition to the enhanced checking, that it
allows us to document higher-order programs in much more accurate way.
We have also implemented several prototypes for operational semantics
with dynamic predprop checking (see the
appendix~\ref{apx:implementation} for a minimalistic implementation),
which we plan to integrate into the already existing assertion
checking mechanisms for first-order assertions.

\bibliographystyle{unsrt}

\input{asrHO_arxiv.bbl}

\appendix

\clearpage
\section{Minimalistic Sample Implementation}
\label{apx:implementation}

The following code (portable to most Prolog systems with minor
changes) shows a minimalistic sample implementation (as an interpreter
\texttt{intr/1}) of the operational semantics for dynamic predprop
checking (Def.~\ref{def:ho-ho-reductions}). Conciseness and simplicity
has been favoured over efficiency. We assume that clauses, assertion
conditions, and predprops have been parsed and stored in \texttt{cl/2},
\texttt{ac/1}, \texttt{pp/2} facts, respectively. The interpreter will  
throw an exception the first time that a failed program assertion is
detected (see \texttt{ext/2} predicate). E.g.,
\texttt{intr((test\_s(1,P),P(1)))} is a valid query while
\texttt{intr((test\_s(1,P),P(-2)))} throws a failed assertion
exception. Predicate \texttt{reset/0} must be called between
\texttt{intr/1} queries to reset error status and temporary
data. In the handler errors can be gathered (as in the semantics) or
execution aborted.

{
\small  
\begin{verbatim}
:- module(_, [reset/0, intr/1], [hiord, dcg, dynamic_clauses]).
:- use_module(library(aggregates)).

% ---------------------------------------------------------------------------
% Sample program data and properties

% negint/1 and nnegint/1 properties
eval_prop(negint(X)) :- integer(X), X < 0.
eval_prop(nnegint(X)) :- integer(X), X >= 0.

% predprops nneg/1 and neg/1
pp(nneg(P), ac(P(X), nneg_c1(P)#success(true, nnegint(X)))).
pp(neg(P), ac(P(X), neg_c1(P)#success(true, negint(X)))).

% assertion conditions and clauses for test_s/2
ac(test_s(N,_P), c1#calls((nnegint(N);negint(N)))).
ac(test_s(N,P), c2#success(nnegint(N), nneg(P))).
ac(test_s(N,P), c3#success(negint(N), neg(P))).
cl(test_s( 1,P), P = z).
cl(test_s(-1,P), P = n).

% clauses for z/1, n/1
cl(z(1), true).  cl(z(-2), true).  cl(n(-1), true).  cl(n(-2), true).

% ---------------------------------------------------------------------------
% Intepreter

:- dynamic hyp_ac/2. % hypothetical assertion condition
:- dynamic negac/1. % (negated) assertion dependency rule

% Reset errors and hypothetical assertion conditions
reset :- retractall(hyp_ac(_, _)), ( retract((negac(_) :- _)), fail ; true ).

% Interpreter with higher-order assertion checking
intr(X) :- ctog(X, X1), !, intr(X1).
intr(X) :- is_blt(X), !, X.
intr((A,B)) :- !, intr(A), intr(B).
intr((A ; B)) :- !, ( intr(A) ; intr(B) ).
intr(A) :-
    get_acs(A, Acs),
    pre(Acs, Ids, []), cl(A, Body), intr(Body), post(Ids, Acs).

% Built-ins
is_blt(true).   is_blt(fail).   is_blt(_ = _).

% From call(N,...) to N(...), where N is a predicate symbol
ctog(X, _) :- var(X), !, throw(inst_error).
ctog(X, X1) :-
    X =.. [call,N|Args],
    ( atom(N) -> true ; throw(inst_error) ),
    X1 =.. [N|Args].

% Get assertion conditions for the given literal A
get_acs(A, Acs) :- ( bagof(Ac, get_ac(A, Ac), Acs) -> true ; Acs = [] ).
get_ac(A, Ac) :- ( ac(A, Ac) ; hyp_ac(A, Ac) ).

pre([]) --> [].   pre([Ac|Acs]) --> pre_(Ac), pre(Acs).
pre_(Id#calls(Pre)) --> { ext(Pre, Id) }.
pre_(Id#success(Pre, _)) --> ( { simp0(Pre, true) } -> [Id] ; [] ).

post([], _Acs).   post([Id|Ids], Acs) :- post_(Id, Acs), post(Ids, Acs).
post_(Id, Acs) :- member(Id0#success(_Pre,Post), Acs), Id == Id0, !, ext(Post, Id).
post_(_, _).

% Check/extend assertion conditions
ext(Props, Id) :-
    simp(Props, Props2), ext_(Props2, Id),
    ( negac(A), atom(A) -> throw(failed_assertion(A)) ; true ).
ext_(true, _Id) :- !.
ext_(false, Id) :- !, assertz((negac(Id) :- true)).
ext_(Props, Id) :- acsubs(Props, Props2), assertz((negac(Id) :- Props2)).

% Add assertion dependency rules
acsubs((A,B), (A2,B2)) :- !, acsubs(A, A2), acsubs(B, B2).
acsubs((A ; B), (A2 ; B2)) :- !, acsubs(A, A2), acsubs(B, B2).
acsubs(ac(L, Id#Ac), negac(Id)) :- ctog(L, L2), assertz(hyp_ac(L2, Id#Ac)).

% Condition simplification
simp(true, R) :- !, R = true.
simp((X;Y), R) :- !, simp(X, Rx), simp(Y, Ry), or(Rx, Ry, R).
simp((X,Y), R) :- !, simp(X, Rx), simp(Y, Ry), and(Rx, Ry, R).
simp(X, R) :- pp(X, Ac), !, R = Ac.
simp(X, R) :- eval_prop(X), !, R = true.
simp(_, R) :- R = false.

% Condition simplification for success preconditions
simp0(true, R) :- !, R = true.
simp0((X,Y), R) :- !, simp0(X, Rx), simp0(Y, Ry), and(Rx, Ry, R).
simp0(X, R) :- eval_prop(X), !, R = true.
simp0(_, R) :- R = false.

or(true, _, R) :- !, R = true.     or(false, X, R) :- !, R = X.
or(_, true, R) :- !, R = true.     or(X, false, R) :- !, R = X.
or(X, Y, (X;Y)).

and(false, _, R) :- !, R = false.  and(true, X, R) :- !, R = X.
and(_, false, R) :- !, R = false.  and(X, true, R) :- !, R = X.
and(X, Y, (X,Y)).
\end{verbatim}
}

\end{document}

%% file: macros.tex
\newcommand{\VS}{\textsf{VS}}
\newcommand{\FS}{\textsf{FS}}
\newcommand{\PS}{\textsf{PS}}

\newcommand{\TS}{\textsf{TS}}
\newcommand{\LS}{\textsf{LS}}

\newcommand{\ar}[1]{\ensuremath{\textsf{ar}(#1)}}

\newcommand{\Q}{\ensuremath{{\cal Q}}}
\newcommand{\A}
  {\ensuremath{{\cal A}}}
\newcommand{\AC}
  {\ensuremath{{\cal A}_{C}}}
\newcommand{\LabAClit}[1]
  {\ensuremath{\AC^{\#}(#1)}}
\newcommand{\ACpp}[1]
  {\ensuremath{\AC[#1]}}
\newcommand{\LabACpp}[1]
  {\ensuremath{\AC^{\#}[#1]}}

\newcommand{\kbd}[1]{\mbox{\tt #1}}
\newcommand{\skbd}[1]{\mbox{\tt\small{#1}}}

\newcommand{\gd}[0]{\mid}
\newcommand{\state}[2]
  {\ensuremath{\langle #1 \gd{} #2 \rangle}}
\newcommand{\exstate}[3]
  {\ensuremath{\langle #1 \gd{} #2 \gd{} #3 \rangle}}

\newcommand{\emptyGoal}{\ensuremath{\square}}

\newcommand{\ADeps}{\ensuremath{{\cal E}}}

\newcommand{\defn}[1]
  {\ensuremath{\textsf{defn}(#1)}}

\newcommand{\answers}[1]{\ensuremath{\textsf{answers}(#1)}}

\newcommand{\currstate}[1]{\ensuremath{{#1}_{[-1]}}}
\newcommand{\nthstate}[2]{\ensuremath{{#1}_{[#2]}}}

\newcommand{\simp}[1]
  {\ensuremath{\textsf{simp}(#1,\theta)}}

\newcommand{\ext}[2]
  {\ensuremath{\textsf{ext}(\A_C,#1,#2)}}

\newcommand{\rtsolve}[2]
  {\textsf{rtsolve}(#1,#2)}

\newcommand{\clause}[2]
  {\ensuremath{#1 \textsf{:-} #2}}

\newcommand{\reduction}[2]
  {\ensuremath{#1 \leadsto #2}}
\newcommand{\reductionStar}[2]
  {\ensuremath{#1 \leadsto^* #2}}

\newcommand{\reductionA}[2]
  {\ensuremath{#1 \leadsto_{\A} #2}}
\newcommand{\reductionStarA}[2]
  {\ensuremath{#1 \leadsto^*_{\A} #2}}

\newcommand{\reductionHAs}[2]
  {\ensuremath{#1 \leadsto_{H\A_{\textsf{\scriptsize{s}}}}#2}}

\newcommand{\reductionHAd}[2]
  {\ensuremath{#1 \leadsto_{H\A_{\textsf{\scriptsize{d}}}} #2}}

\newcommand{\derivations}[1]
  {\ensuremath{\textsf{derivs}(#1)}}

\newcommand{\derivationsA}[1]
  {\ensuremath{\textsf{derivs}_{\A}(#1)}}

\newcommand{\derivationsHAd}[1]
  {\ensuremath{\textsf{derivs}_{H\A_{\textsf{\scriptsize{d}}}}(#1)}}

\newcommand{\asrId}[1]{\ensuremath{c_{#1}}}
\newcommand{\negAsrId}[1]{\ensuremath{\bar{c}_{#1}}}

\newcommand{\hypId}[1]{\ensuremath{h_{#1}}}
\newcommand{\negHypId}[1]{\ensuremath{\bar{h}_{#1}}}

\newcommand{\checkedAsr}[1]{\ensuremath{\textsf{checked}(#1)}}
\newcommand{\falseAsr}[1]{\ensuremath{\textsf{false}(#1)}}

\newcommand{\checkLitLab}[1]{\ensuremath{\textsf{check}(#1)}}

\newcommand{\callsAsr}[2]{\ensuremath{\textsf{calls}(#1, #2)}}

\newcommand{\successAsr}[3]{\ensuremath{\textsf{success}(#1, #2, #3)}}

\newcommand{\labCallsAsr}[3]
  {\ensuremath{\textsf{\asrId{#1}\#calls}(#2, #3)}}
\newcommand{\labSuccessAsr}[4]
  {\ensuremath{\textsf{\asrId{#1}\#success}(#2, #3, #4)}}

\newcommand{\hypSuccessAsr}[4]
  {\ensuremath{\textsf{\hypId{#1}\#success}(#2, #3, #4)}}

\newcommand{\errorErase}[1]
  {\ensuremath{#1^{\circ}}}